\documentclass[11pt]{article}

\usepackage{amsmath,amssymb,amsfonts,amsthm,mathtools}
\usepackage{geometry}
\usepackage{microtype}
\usepackage{enumitem}
\usepackage[colorlinks=true,linkcolor=blue,citecolor=blue,urlcolor=blue]{hyperref}

\theoremstyle{plain}
\newtheorem{theorem}{Theorem}[section]
\newtheorem{lemma}[theorem]{Lemma}
\newtheorem{proposition}[theorem]{Proposition}
\newtheorem{corollary}[theorem]{Corollary}
\newtheorem{conjecture}[theorem]{Conjecture}

\theoremstyle{definition}
\newtheorem{definition}[theorem]{Definition}
\newtheorem{remark}[theorem]{Remark}
\newtheorem{example}[theorem]{Example}

\newcommand{\N}{\mathbb{N}}
\newcommand{\Z}{\mathbb{Z}}
\newcommand{\PP}{\mathbb{P}} % primes (as a set)
\newcommand{\cP}{\mathcal{P}} % prime sets in intervals
\newcommand{\cC}{\mathcal{C}} % constellation events
\newcommand{\cS}{\mathcal{S}} % sifted sets
\newcommand{\cJ}{\mathcal{J}} % Selberg quadratic form

\newcommand{\1}{\mathbf{1}} % indicator

\DeclareMathOperator{\supp}{supp}

\newcommand{\E}{\mathbb{E}} % expectation

\newcommand{\logtwo}{\log_2}

\newcommand{\vp}{\varphi}
\newcommand{\Lam}{\Lambda}

\newcommand{\K}{K}
\newcommand{\code}{\mathrm{code}}

\newcommand{\WeakLoc}{\mathrm{Weak}_{\mathrm{loc}}}
\newcommand{\WeakSel}{\mathrm{Weak}_{\mathrm{Sel}}}

\title{Prime Successor Irreducibility:\\
Turing Machine Complexity, Kolmogorov Complexity,\\
and Weakness-Based Formulations}
\author{Ben Goertzel \and Bill Lauritzen}
\date{\today}

\begin{document}
\maketitle

\begin{abstract}
We develop a family of conjectures and theorems expressing the idea that the sequence of prime numbers exhibits a form of computational irreducibility in the transition from one prime to its successor. Informally, given a prime $p$, there is no general algorithm that can compute the least prime greater than $p$ substantially faster than sequentially testing candidate integers for primality, except possibly on sparse sets of inputs.

Our framework proceeds along several complementary lines. First, we formalize Prime Successor Irreducibility in a uniform Turing-machine complexity model (PSI-T), asserting lower bounds on the running time of any algorithm computing prime successors relative to a natural sequential baseline. Second, we propose a Kolmogorov-complexity formulation (PSI-K), which asserts that typical prime gaps are algorithmically incompressible given the scale at which they occur; we prove PSI-K$(c,\delta)$ unconditionally for all fixed $c<1$ using standard sieve upper bounds. Third, we develop weakness-based formulations: PSI-W (sparse-set anti-concentration) shows that no small ``menu'' of candidate gap values captures a noticeable fraction of primes, while PSI-W-LE (logical entropy) shows that collision probabilities decay and logical entropy tends to $1$. These results extend to prime constellations and consecutive gap vectors. Finally, we present a sieve-theoretic framework that connects local obstruction patterns to Selberg weakness parameters, providing clean compositional bounds.

The PSI-K and weakness formulations connect irreducibility directly to classical statistical questions about prime gaps. Using the relationship between Kolmogorov complexity and Shannon entropy, we derive rigorous lower bounds on the entropy of prime gap distributions in dyadic intervals $[X,2X]$. Together, these formulations provide a unified complexity-theoretic perspective on the apparent local unpredictability of the prime sequence, without asserting randomness or independence.
\end{abstract}

\tableofcontents

% ============================================================
\section{Introduction}
% ============================================================

Let $p_n$ denote the $n$-th prime and $p_{n+1}$ its successor. A basic computational task in number theory and cryptography is: given a prime $p$, compute the least prime $p' > p$. The most direct method simply tests each integer $p+1, p+2, \ldots$ for primality until a prime is found.

The informal idea underlying Prime Successor Irreducibility (PSI) is that this na\"ive approach is not merely a baseline but is essentially optimal, up to polylogarithmic factors, for large primes. That is, there is no uniform algorithm that can systematically ``jump'' to the next prime while avoiding computational work comparable to examining almost all integers in the intervening gap.

A version of this idea was articulated by Lauritzen as a standalone conjecture emphasizing the absence of low-information shortcuts from one prime to its successor \cite{Lauritzen}. The present paper embeds that intuition into a more formal complexity-theoretic framework and develops its consequences.

Importantly, PSI does not assert that primes are random, independent, or patternless. Rather, it asserts a specific negative claim: that no uniform, bounded-description algorithm can exploit hidden structure to compute prime successors asymptotically faster than sequential search on almost all inputs.

\subsection{Overview of formulations}

We present several complementary formalizations of this idea:

\begin{enumerate}[leftmargin=2em]
\item \textbf{PSI-T (Turing-machine formulation).} We express irreducibility as lower bounds on the running time of any uniform prime-successor algorithm in the standard bit-complexity model. We state four variants that differ in how often the conjectured lower bound is required to hold: infinitely many primes, positive density, density one, or all but finitely many.

\item \textbf{PSI-K (Kolmogorov-complexity formulation).} We express irreducibility as incompressibility of typical prime gaps conditional on their scale, and derive entropy bounds for gap distributions. A key contribution of this paper is an \emph{unconditional proof} that PSI-K$(c,\delta)$ holds for every fixed $c < 1$ and every $\delta \in (0,1)$, using only standard Selberg/Brun-type sieve upper bounds.

\item \textbf{PSI-W (sparse-set anti-concentration).} We prove that no small ``menu'' of candidate gap values---specifically, any set of size at most $(\log X)^c$ for $c < 1$---can capture a noticeable fraction of primes in $[X,2X]$. This is the set-valued relaxation of PSI-K.

\item \textbf{PSI-W-LE (logical entropy).} We show that the collision probability of the gap distribution tends to $0$, hence logical entropy tends to $1$. This provides a clean measure of diversity that tensorizes under independent sampling.

\item \textbf{Sieve-theoretic weakness framework.} We develop a notion of ``local weakness'' for prime-counting problems defined by congruence obstructions, showing how local weakness parameters canonically induce Selberg majorants and yield clean compositional bounds.
\end{enumerate}

\subsection{Key results}

The main contributions of this paper include:

\begin{itemize}[leftmargin=2em]
\item An unconditional proof of PSI-K$(c,\delta)$ for all $c < 1$ (Theorem~\ref{thm:PSIK-unconditional}).
\item Menu bounds showing $\Pr[G_X \in S] \ll |S| \log\log X / \log X$ for any finite set $S$ (Theorem~\ref{thm:menu-single}).
\item Collision probability bounds $w_X \ll \log\log X / \log X$ implying logical entropy $\ell_X \to 1$ (Theorem~\ref{thm:collision-single}).
\item Extensions to prime constellations and consecutive gap vectors (Section~\ref{sec:constellations}).
\item A sieve-theoretic framework connecting local obstruction patterns to prime-counting upper bounds (Section~\ref{sec:sieve-framework}).
\end{itemize}

\subsection{Relationship between formulations}

PSI-T and PSI-K do not formally imply one another, but they rule out the same underlying phenomenon: the existence of a low-complexity rule or bounded-information functional that maps a prime to its successor. PSI-T expresses this in terms of unavoidable computational work; PSI-K expresses it in terms of irreducible information content. The weakness formulations PSI-W and PSI-W-LE provide provable statistical consequences that follow from standard analytic number theory.

Together these formulations provide complementary views of the same irreducibility intuition, with the weakness results serving as unconditionally provable proxies for the stronger conjectural statements.

% ============================================================
\section{Preliminaries}
\label{sec:preliminaries}
% ============================================================

\subsection{Primes and prime gaps}

Let $\{p_n\}_{n \geq 1}$ denote the increasing sequence of all primes. Thus $p_1 = 2$, $p_2 = 3$, $p_3 = 5$, and so on. We define the \emph{prime gap} after $p_n$ by
\[
g_n := p_{n+1} - p_n.
\]

For a prime $p$ in a dyadic interval $[X, 2X]$, we write $g(p) = p' - p$, where $p'$ is the least prime greater than $p$. We let $p^+$ denote the next prime after $p$, and more generally $p^{+(j)}$ denotes the $j$-th successor prime.

We write $\ell(p) = \lfloor \log_2 p \rfloor + 1$ for the bit-length of $p$ in binary.

\subsection{Dyadic intervals and the gap distribution}

Fix $X \geq 2$ and let
\[
\cP_X := \PP \cap [X, 2X], \qquad \pi_X := |\cP_X| = \pi(2X) - \pi(X).
\]
We sample a prime $P$ uniformly from $\cP_X$ and define the gap random variable
\[
G_X := g(P).
\]
Let $\mu_X$ denote the probability mass function of $G_X$:
\[
\mu_X(h) := \Pr[G_X = h], \qquad h \in \N.
\]

By Bertrand's postulate, for $p \in [X, 2X]$ we have $p^+ \leq 2p \leq 4X$, hence
\begin{equation}\label{eq:gap-range}
G_X \leq 2X \quad \text{always.}
\end{equation}

\subsection{Computational model and uniformity}

We work in the standard deterministic Turing-machine model of bit complexity. All algorithms are required to be uniform Turing machines with a fixed finite description. In particular:
\begin{itemize}[leftmargin=2em]
\item No oracle access is permitted.
\item No non-uniform advice depending on the input size is permitted.
\item The algorithm's description must be independent of the input prime.
\end{itemize}
This excludes models in which large precomputed tables, size-indexed advice strings, or external black-box information trivialize the successor problem.

\subsection{Kolmogorov complexity}

Fix a universal prefix-free Turing machine $U$. For a finite binary string $x$, the (plain) Kolmogorov complexity of $x$ is
\[
\K(x) := \min\{|p| : U(p) = x\},
\]
where $|p|$ denotes the length of program $p$ in bits. Conditional complexity $\K(x \mid y)$ is defined analogously.

We fix an integer-encoding map $\code(n) \in \{0,1\}^*$ such that $|\code(n)| = \log_2 n + O(1)$.

\textbf{Logarithms.} We use $\logtwo$ for base-2 logarithms (bits), and $\log$ for the natural logarithm. Since $\logtwo X = (\log X)/(\log 2)$, any appearance of $(\logtwo X)^c$ may be replaced by a constant multiple of $(\log X)^c$.

\subsection{Shannon entropy and its relation to Kolmogorov complexity}

Let $P$ be a probability distribution on a countable set $\mathcal{X} \subseteq \{0,1\}^*$. The Shannon entropy of $P$ is
\[
H(P) := -\sum_{x \in \mathcal{X}} P(x) \log_2 P(x).
\]

For a recursive probability distribution $P$, we have
\[
0 \leq \E_{x \sim P}[\K(x)] - H(P) \leq \K(P),
\]
where $\K(P)$ is the Kolmogorov complexity of the distribution $P$ itself.

\subsection{Quantale weakness}
\label{subsec:quantale-weakness}

Several formulations in this paper measure ``weakness'' of hypotheses or distributions---a concept borrowed from cognitive science and AGI theory \cite{bennett-thesis, goertzel-weakness}. This subsection provides the general mathematical framework; subsequent sections instantiate it for specific number-theoretic applications.

\subsubsection{Commutative quantales}

A \emph{commutative quantale} $(V, \otimes, \oplus, e, \leq)$ is a complete lattice $(V, \leq)$ equipped with a monoidal product $\otimes$ (associative, commutative, with unit $e$), where the join $\oplus$ is the lattice supremum. The product $\otimes$ distributes over arbitrary joins. Intuitively, $\otimes$ combines ``intensities'' while $\oplus$ aggregates or joins multiple cases.

\begin{example}[Key quantale examples]\label{ex:quantales}
\begin{enumerate}[leftmargin=2em]
\item \textbf{Probabilistic quantale:} $([0,1], \times, +, 1, \leq)$, where $\otimes = \times$ (multiplication) and $\oplus = +$ (truncated addition or max). This underlies collision probability and logical entropy.
\item \textbf{Counting/MDL quantale:} $(\N \cup \{\infty\}, +, \max, 0, \leq)$, where $\otimes = +$ (addition) and $\oplus = \max$. This connects to description length and Kolmogorov complexity.
\item \textbf{Set-theoretic quantale:} $(\mathcal{P}(W), \cap, \cup, W, \subseteq)$ for a universe of possible worlds $W$. Here $\otimes = \cap$ (intersection) and $\oplus = \cup$ (union). This is the setting for Bennett's original notion of weakness.
\end{enumerate}
\end{example}

\subsubsection{Quantale weakness on relations}

Let $\mathcal{C}$ be enriched over a commutative quantale $(V, \otimes, \oplus, e, \leq)$ with a universe of objects $U$ and a valuation $\mu : U \to V$. For a relation (or hypothesis) $H \subseteq U \times U$ representing pairs that the hypothesis \emph{fails to distinguish}, define the \textbf{quantale weakness}:
\begin{equation}\label{eq:quantale-weakness}
w(H) := \bigoplus_{(u,v) \in H} \bigl(\mu(u) \otimes \mu(v)\bigr).
\end{equation}

\textbf{Interpretation.} The weakness $w(H)$ measures how many (and how important) pairwise distinctions the hypothesis $H$ fails to make. Pairs $(u,v)$ with large $\mu(u) \otimes \mu(v)$ contribute more; the join $\oplus$ aggregates all such undistinguished pairs. A hypothesis with high weakness ``commits to less'' and is thus simpler or more general in the Occam's Razor sense.

This framework unifies several classical notions:
\begin{itemize}[leftmargin=2em]
\item \textbf{Bennett/MDL weakness:} In the counting quantale, if $\mu(u) = 1$ for each element (token), then $w(H)$ counts the number of undistinguished pairs, which relates to description length via $\text{MDL}(H) \approx \log_2 N - \log_2 w(H)$ for $N$ total possibilities.
\item \textbf{Probabilistic weakness:} In the probabilistic quantale with $\mu(u) = p(u)$ a probability distribution, the weakness of the same-block relation $H_\pi$ of a partition $\pi$ becomes
\[
w(H_\pi) = \sum_{(u,v) \in H_\pi} p(u) p(v),
\]
which is exactly the \emph{collision probability}. Its complement $1 - w(H_\pi)$ is Ellerman's \emph{logical entropy}.
\item \textbf{Set-theoretic weakness:} In the powerset quantale, $w(H) = \bigcup_{(u,v) \in H} (\mu(u) \cap \mu(v))$ measures the set of possible worlds left undistinguished.
\end{itemize}

\subsubsection{Connection to this paper}

The weakness notions introduced in Sections~\ref{sec:PSI-W}, \ref{sec:PSI-W-LE}, and \ref{sec:sieve-framework} are all instances of quantale weakness:
\begin{itemize}[leftmargin=2em]
\item \textbf{PSI-W (menu bounds)} measures how much probability mass can be captured by a small set of gap values---essentially bounding $w(S)$ for small $|S|$ in the probabilistic quantale.
\item \textbf{PSI-W-LE (logical entropy weakness)} directly computes the collision probability $w_X = \sum_h \mu_X(h)^2$, which is quantale weakness in the probabilistic quantale where $H$ is the diagonal relation (same-gap pairs).
\item \textbf{Local and Selberg weakness} (Section~\ref{sec:sieve-framework}) use a multiplicative/Euler-product structure that can be viewed as quantale weakness in a suitable product quantale, where the valuation encodes local density reductions from sieving.
\end{itemize}

This unified perspective clarifies why these apparently different ``weakness'' measures share common algebraic properties such as compositionality and monotonicity.

% ============================================================
\section{PSI-T: Turing Machine Formulation}
\label{sec:PSI-T}
% ============================================================

\subsection{The sequential prime successor algorithm}

We define the sequential successor algorithm as our baseline:

\begin{description}
\item[Input:] A prime $p$ written in binary.
\item[Algorithm:] Let $m \leftarrow p+1$. Repeat: test whether $m$ is prime using a fixed polynomial-time primality test with cost $C(\cdot)$. If prime, halt and output $m$. Otherwise increment $m$ and continue.
\item[Output:] The least prime $m > p$.
\end{description}

We call this algorithm $\mathrm{Seq}$. For each prime $p_n$, define $T_{\mathrm{Seq}}(p_n)$ to be the number of Turing steps performed on input $p_n$. Since on input $p_n$ the algorithm tests exactly the $g_n$ integers $p_n+1, \ldots, p_{n+1}$, each with bit-length $\ell(n) + O(1)$, we have
\[
T_{\mathrm{Seq}}(p_n) \asymp g_n \cdot C(\ell(p_n)),
\]
where $C(k)$ is the worst-case cost of primality testing on $k$-bit inputs.

\begin{remark}[On the choice of primality test]
We fix $C(k)$ to be the cost of some deterministic polynomial-time primality test, such as AKS or Miller-Rabin with deterministic witness selection. Concretely, we require $C(k) \leq k^d$ for some constant $d \geq 1$ (see Appendix~\ref{app:primality} for explicit estimates). 

The PSI-T conjectures concern whether one can avoid sequential search through the gap, \emph{not} whether one can improve the primality test itself. An algorithm that ``beats'' $\mathrm{Seq}$ by using a faster primality subroutine is not a counterexample to PSI-T; rather, it simply changes the baseline. The content of PSI-T is that, for any fixed polynomial-time primality test, one cannot systematically skip over most candidates in the gap---the computational work scales with the gap size $g_n$, not merely with the bit-length $\ell(p_n)$.

Put differently: PSI-T asserts that the number of primality tests performed is $\Omega(g_n / (\log p_n)^\alpha)$ for some $\alpha > 0$, regardless of which polynomial-time test is used. The irreducibility claim is about the search process, not the per-candidate test.
\end{remark}

\subsection{Prime successor algorithms}

\begin{definition}[Prime successor algorithm]
A Turing machine $M$ is called a \emph{prime successor algorithm} if, on input the binary representation of a prime $p$, it halts and outputs the binary representation of the least prime strictly greater than $p$.
\end{definition}

\subsection{PSI-T conjectures}

We state four increasingly strong versions of PSI-T, differing only in how frequently the lower bound must hold.

\begin{conjecture}[PSI-T, infinitely many primes]\label{conj:PSIT-inf}
Let $M$ be any prime successor algorithm. Then there exists a constant $\alpha > 0$ and infinitely many indices $n$ such that
\[
\mathrm{time}_M(p_n) \geq \frac{1}{(\log p_n)^\alpha} T_{\mathrm{Seq}}(p_n).
\]
\end{conjecture}

\begin{conjecture}[PSI-T, positive density]\label{conj:PSIT-posdens}
Let $M$ be any prime successor algorithm. Then there exist constants $\alpha > 0$, $c > 0$, and a subset $S$ of the primes with natural density $\delta_{\PP}(S) = c > 0$, such that for all $p_n \in S$,
\[
\mathrm{time}_M(p_n) \geq \frac{1}{(\log p_n)^\alpha} T_{\mathrm{Seq}}(p_n).
\]
\end{conjecture}

\begin{conjecture}[PSI-T, density one]\label{conj:PSIT-dens1}
Let $M$ be any prime successor algorithm. Then there exists a constant $\alpha > 0$ and a subset $S$ of the primes with natural density $\delta_{\PP}(S) = 1$ such that for all $p_n \in S$,
\[
\mathrm{time}_M(p_n) \geq \frac{1}{(\log p_n)^\alpha} T_{\mathrm{Seq}}(p_n).
\]
\end{conjecture}

\begin{conjecture}[PSI-T, almost everywhere]\label{conj:PSIT-ae}
Let $M$ be any prime successor algorithm. Then there exists a constant $\alpha > 0$ and an integer $N \geq 1$ such that for all $n \geq N$,
\[
\mathrm{time}_M(p_n) \geq \frac{1}{(\log p_n)^\alpha} T_{\mathrm{Seq}}(p_n).
\]
\end{conjecture}

Each conjecture implies the previous ones in the list.

\begin{remark}[Interpretation]
PSI-T does not forbid occasional shortcuts or unusually small gaps. It asserts only that no uniform algorithm can asymptotically outperform the sequential baseline by more than polylogarithmic factors on almost all inputs.
\end{remark}

% ============================================================
\section{PSI-K: Kolmogorov Complexity Formulation}
\label{sec:PSI-K}
% ============================================================

\subsection{Motivation}

An alternative expression of irreducibility is informational rather than operational. Instead of lower-bounding running time, we ask whether the prime gap itself contains irreducible information once trivial scale information is supplied.

Near scale $X$, typical prime gaps are of order $\log X$, and their binary encodings have length $\Theta(\logtwo \log X)$. PSI-K conditions on the scale $X$ to factor out this trivial information.

\subsection{The PSI-K conjecture}

\begin{definition}[PSI-K$(c,\delta)$]
Let $c > 0$ and $0 < \delta < 1$. We say that \emph{PSI-K$(c,\delta)$ holds} if there exists $X_0$ such that for all $X \geq X_0$, at least a proportion $1 - \delta$ of primes $p \in [X, 2X]$ satisfy
\[
\K(\code(g(p)) \mid X) \geq c \logtwo \logtwo X.
\]
\end{definition}

\begin{conjecture}[PSI-K]\label{conj:PSI-K}
There exist constants $c > 0$ and $\delta < 1$ such that PSI-K$(c,\delta)$ holds.
\end{conjecture}

\subsection{Entropy consequences}

Using the relationship between Kolmogorov complexity and Shannon entropy, PSI-K implies lower bounds on the entropy of the prime gap distribution in dyadic intervals. Specifically, under PSI-K with parameters $(c, \delta)$, for all sufficiently large $X$:
\[
H(\mu_X) \geq (1 - \delta) c \logtwo \logtwo X - O(1).
\]

Moreover, the pointwise inequality $\K(g \mid X) \leq -\logtwo \mu_X(g) + O(1)$ combined with PSI-K implies that each ``typical'' gap value (one that is incompressible given $X$) can account for at most $O((\log X)^{-c})$ of the total probability mass.

% ============================================================
\section{Unconditional Proof of PSI-K for $c < 1$}
\label{sec:PSIK-proof}
% ============================================================

We now prove that PSI-K$(c,\delta)$ holds unconditionally for every fixed $c \in (0,1)$ and every $\delta \in (0,1)$. The proof combines Kolmogorov counting with standard sieve upper bounds.

\subsection{Kolmogorov counting}

For $X \geq 3$ and $c > 0$, define the set of low-complexity gap values:
\[
S_c(X) := \{g \in \N : \K(\code(g) \mid X) < c \logtwo \logtwo X\}.
\]

\begin{lemma}[Kolmogorov counting bound]\label{lem:kolmogorov-count}
For every $c > 0$ and all $X \geq 3$,
\[
|S_c(X)| \leq 2^{c \logtwo \logtwo X} = (\logtwo X)^c \ll (\log X)^c.
\]
\end{lemma}

\begin{proof}
For each $g \in S_c(X)$ there exists a binary program string $\tau$ of length $|\tau| < c \logtwo \logtwo X$ such that $U(\tau, X) = \code(g)$. For fixed $X$, each program produces at most one output, hence the number of such outputs is at most the number of binary strings of length $< c \logtwo \logtwo X$, which is $< 2^{c \logtwo \logtwo X} = (\logtwo X)^c$.
\end{proof}

\subsection{Sieve upper bound for prime pairs}

For an integer $h \geq 1$ and parameter $X \geq 2$, define
\[
A_h(X) := \#\{p \in [X, 2X] \cap \PP : p + h \in \PP\}.
\]

\begin{theorem}[Selberg/Brun upper bound for prime pairs]\label{thm:pair-sieve}
There exist absolute constants $C_{\mathrm{pair}} \geq 1$ and $X_1 \geq 2$ such that for all $X \geq X_1$ and all even integers $h$ with $1 \leq h \leq 2X$,
\[
A_h(X) \leq C_{\mathrm{pair}} F(h) \frac{X}{(\log X)^2},
\]
where
\[
F(h) := \prod_{\substack{q \mid h \\ q > 2}} \frac{q-1}{q-2}.
\]
Moreover, $F(h) \ll \log \log h$ for $h \geq 3$.
\end{theorem}

\begin{lemma}[$F(h)$ is $O(\log \log h)$]\label{lem:F-bound}
There exist absolute constants $C_F \geq 1$ and $h_0 \geq 3$ such that for all integers $h \geq h_0$,
\[
F(h) \leq C_F \log \log h.
\]
In particular, for all $h$ with $2 \leq h \leq 2X$ and $X$ sufficiently large, $F(h) \leq C_F \log \log X$.
\end{lemma}

\begin{proof}
We have
\[
F(h) = \prod_{\substack{q \mid h \\ q > 2}} \frac{q-1}{q-2} \leq C_* \prod_{q \mid h} \frac{q}{q-1} = C_* \frac{h}{\vp(h)},
\]
where $C_* < \infty$ is a convergent product constant. The classical bound $h/\vp(h) \ll \log \log h$ completes the proof.
\end{proof}

\subsection{From prime pairs to successor gaps}

\begin{lemma}[Successor gaps dominated by prime pairs]\label{lem:gap-pair}
For every $g \geq 1$ and $X \geq 2$,
\[
\#\{p \in [X, 2X] \cap \PP : g(p) = g\} \leq A_g(X).
\]
\end{lemma}

\begin{proof}
If $g(p) = g$ then $p$ and $p + g = p^+$ are both prime.
\end{proof}

\begin{lemma}[Gap range bound]\label{lem:gap-range-bound}
For $X \geq 3$ and every prime $p \in [X, 2X]$, we have $g(p) < p \leq 2X$ and $g(p)$ is even (for $p \geq 3$).
\end{lemma}

\subsection{Main theorem}

Define the exceptional primes:
\[
P_{\mathrm{low}}(X; c) := \{p \in \cP_X : g(p) \in S_c(X)\}.
\]

\begin{proposition}[Main bound on the exceptional fraction]\label{prop:exceptional-bound}
Fix $c \in (0,1)$. There exist constants $C_0 \geq 1$ and $X_2$ such that for all $X \geq X_2$,
\[
\frac{|P_{\mathrm{low}}(X; c)|}{\pi_X} \leq C_0 \frac{\log \log X}{(\log X)^{1-c}}.
\]
\end{proposition}

\begin{proof}
Partition by gap value:
\[
|P_{\mathrm{low}}(X; c)| = \sum_{g \in S_c(X)} \#\{p \in [X, 2X] \cap \PP : g(p) = g\} \leq \sum_{g \in S_c(X)} A_g(X).
\]
Only gaps $g \leq 2X$ can occur, so we restrict to $S_c(X) \cap [1, 2X]$. Applying Theorem~\ref{thm:pair-sieve}:
\[
|P_{\mathrm{low}}(X; c)| \ll \frac{X}{(\log X)^2} \sum_{g \in S_c(X) \cap [1, 2X]} F(g).
\]
By Lemmas~\ref{lem:kolmogorov-count} and \ref{lem:F-bound}:
\[
\sum_{g \in S_c(X) \cap [1, 2X]} F(g) \leq |S_c(X)| \cdot C_F \log \log X \ll (\log X)^c \log \log X.
\]
Thus
\[
|P_{\mathrm{low}}(X; c)| \ll X \frac{(\log X)^c \log \log X}{(\log X)^2} = X \frac{\log \log X}{(\log X)^{2-c}}.
\]
Using $\pi_X \gg X / \log X$:
\[
\frac{|P_{\mathrm{low}}(X; c)|}{\pi_X} \ll \frac{\log \log X}{(\log X)^{1-c}}.
\]
\end{proof}

\begin{theorem}[PSI-K$(c,\delta)$ for all $c < 1$]\label{thm:PSIK-unconditional}
Fix $c \in (0,1)$ and $\delta \in (0,1)$. Then there exists $X_0$ such that for all $X \geq X_0$, at least a proportion $1 - \delta$ of primes $p \in [X, 2X]$ satisfy
\[
\K(\code(g(p)) \mid X) \geq c \logtwo \logtwo X.
\]
Equivalently, PSI-K$(c,\delta)$ holds for every fixed $c < 1$.
\end{theorem}

\begin{proof}
By Proposition~\ref{prop:exceptional-bound}, the fraction of exceptional primes is at most $C_0 (\log \log X)/(\log X)^{1-c}$. Since $1 - c > 0$, this tends to $0$ as $X \to \infty$. Choose $X_0$ such that this fraction is at most $\delta$ for all $X \geq X_0$.
\end{proof}

\begin{remark}[Why this argument does not reach $c = 1$]
The proof uses a union bound over at most $(\log X)^c$ candidate gap values. When $c \uparrow 1$, the bound approaches about $\log X$ candidates, and the sieve bound for each fixed gap is too weak to force the union bound to vanish.
\end{remark}

% ============================================================
\section{PSI-W: Sparse-Set Anti-Concentration}
\label{sec:PSI-W}
% ============================================================

We now formulate and prove ``weakness'' relaxations of PSI-K that provide clean anti-concentration and entropy bounds. As described in Section~\ref{subsec:quantale-weakness}, the term ``weakness'' originates from its use in cognitive science and AGI theory \cite{bennett-thesis, goertzel-weakness}, where it measures how much a representation ``commits to''---weak representations make minimal assumptions and thus apply broadly.

In the quantale-weakness framework, PSI-W corresponds to bounding the weakness of small candidate sets in the probabilistic quantale $([0,1], \times, +, 1, \leq)$. Specifically, for a finite set $S$ of gap values, we measure the ``coverage weakness''
\[
w(S) := \sum_{h \in S} \mu_X(h) = \Pr[G_X \in S],
\]
which is a degenerate case of formula~\eqref{eq:quantale-weakness} where the relation $H$ consists of all pairs $(h,h)$ for $h \in S$ (the diagonal), and the valuation is $\mu(h) = \sqrt{\mu_X(h)}$. The key result is that this weakness is small whenever $|S|$ is subpolynomial in $\log X$.

\subsection{Definition of PSI-W}

\begin{definition}[PSI-W: sparse-set anti-concentration]\label{def:psi-w}
Fix $c \in (0,1)$. We say \emph{PSI-W$(c)$ holds} if
\[
\sup_{\substack{S \subseteq \N \\ |S| \leq (\log X)^c}} \Pr[G_X \in S] \xrightarrow[X \to \infty]{} 0.
\]
\end{definition}

\subsection{Bounds on individual gap probabilities}

\begin{lemma}[Each fixed gap value is rare]\label{lem:single-atom}
There is an absolute constant $C > 0$ such that for all $X \geq 3$ and all $1 \leq h \leq 2X$,
\[
\mu_X(h) \leq C \frac{\log \log(3X)}{\log X}.
\]
\end{lemma}

\begin{proof}
If $g(p) = h$ then $p$ and $p + h$ are both prime, hence
\[
\#\{p \in \cP_X : g(p) = h\} \leq A_h(X).
\]
Therefore $\mu_X(h) \leq A_h(X) / \pi_X$. Using Theorem~\ref{thm:pair-sieve} and $\pi_X \gg X / \log X$:
\[
\mu_X(h) \ll \frac{F(h)}{\log X} \ll \frac{\log \log(3X)}{\log X}.
\]
\end{proof}

\subsection{Menu bound}

\begin{theorem}[Menu bound for gaps]\label{thm:menu-single}
There is an absolute constant $C > 0$ such that for all $X \geq 3$ and all finite sets $S \subseteq \{1, 2, \ldots, 2X\}$,
\[
\Pr[G_X \in S] \leq C \frac{|S| \log \log(3X)}{\log X}.
\]
\end{theorem}

\begin{proof}
By union bound and Lemma~\ref{lem:single-atom}:
\[
\Pr[G_X \in S] = \sum_{h \in S} \mu_X(h) \leq |S| \cdot \max_{h \in S} \mu_X(h) \leq C \frac{|S| \log \log(3X)}{\log X}.
\]
\end{proof}

\begin{corollary}[PSI-W$(c)$ for every fixed $c < 1$]\label{cor:psi-w}
Fix $c \in (0,1)$. Then PSI-W$(c)$ holds.
\end{corollary}

\begin{proof}
If $|S| \leq (\log X)^c$, then Theorem~\ref{thm:menu-single} gives
\[
\Pr[G_X \in S] \ll (\log X)^{c-1} \log \log(3X) \xrightarrow[X \to \infty]{} 0.
\]
\end{proof}

\begin{remark}
PSI-W is the ``set-valued'' relaxation of PSI-K: instead of saying the gap is hard to \emph{name} with few bits, it says the gap is hard to \emph{cover} by any short list of candidate values.
\end{remark}

% ============================================================
\section{PSI-W-LE: Logical Entropy and Collision Probability}
\label{sec:PSI-W-LE}
% ============================================================

\subsection{Logical entropy as quantale weakness}

The logical entropy formulation provides the most direct instantiation of quantale weakness in this paper. Recall from Section~\ref{subsec:quantale-weakness} that in the probabilistic quantale $([0,1], \times, +, 1, \leq)$ with valuation $\mu(h) = \mu_X(h)$ (the gap probability), the quantale weakness of the ``same-value'' relation $H_= = \{(h,h) : h \in \supp(\mu_X)\}$ is
\[
w(H_=) = \sum_{h} \mu_X(h) \otimes \mu_X(h) = \sum_{h} \mu_X(h)^2,
\]
which is exactly the \emph{collision probability}. The complement $1 - w(H_=)$ is Ellerman's \emph{logical entropy} \cite{goertzel-weakness}, which measures the probability that two independent samples from the distribution are \emph{distinct}.

\subsection{Definition}

\begin{definition}[PSI-W-LE: logical entropy weakness]\label{def:psi-w-le}
Let $\mu_X$ be the gap distribution. Define the \emph{collision probability} (quantale weakness of the identity relation)
\[
w_X := \sum_{h \in \N} \mu_X(h)^2 = \Pr[G_X = G_X'],
\]
where $G_X'$ is an independent copy of $G_X$. Define the \emph{logical entropy}
\[
\ell_X := 1 - w_X.
\]
We say \emph{PSI-W-LE holds} if $w_X \to 0$ (equivalently $\ell_X \to 1$) as $X \to \infty$.
\end{definition}

The quantity $w_X$ is the quantale weakness from formula~\eqref{eq:quantale-weakness} applied to the diagonal relation in the probabilistic quantale. High logical entropy $\ell_X \to 1$ means the gap distribution becomes maximally ``spread out'' in the sense that random collisions become vanishingly rare.

\subsection{Collision probability decay}

\begin{theorem}[Collision probability decays]\label{thm:collision-single}
There is an absolute constant $C > 0$ such that for all $X \geq 3$,
\[
w_X = \sum_{h \in \N} \mu_X(h)^2 \leq C \frac{\log \log(3X)}{\log X}.
\]
Equivalently, the logical entropy satisfies
\[
\ell_X = 1 - w_X \geq 1 - C \frac{\log \log(3X)}{\log X}.
\]
\end{theorem}

\begin{proof}
Since $\sum_h \mu_X(h) = 1$ and $\mu_X(h) \geq 0$:
\[
w_X = \sum_h \mu_X(h)^2 \leq \left(\max_h \mu_X(h)\right) \sum_h \mu_X(h) = \max_h \mu_X(h).
\]
Apply Lemma~\ref{lem:single-atom}.
\end{proof}

\begin{corollary}[Shannon entropy grows like $\log \log X$]\label{cor:entropy-single}
Let $H(G_X)$ be the Shannon entropy of $G_X$ in bits. Then
\[
H(G_X) \geq -\logtwo(w_X) \geq \logtwo \log X - O(\logtwo \log \log X).
\]
\end{corollary}

\begin{proof}[Proof sketch]
The R\'enyi entropy of order 2 is $H_2(G_X) := -\logtwo(\sum_h \mu_X(h)^2) = -\logtwo(w_X)$. A standard inequality states that Shannon entropy dominates R\'enyi-2 entropy: $H(G_X) \geq H_2(G_X) = -\logtwo(w_X)$. 

By Theorem~\ref{thm:collision-single}, $w_X \leq C (\log \log(3X))/(\log X)$ for some constant $C$. Thus
\[
-\logtwo(w_X) \geq -\logtwo\left(\frac{C \log \log(3X)}{\log X}\right) = \logtwo(\log X) - \logtwo(C \log \log(3X)),
\]
which equals $\logtwo \log X - O(\logtwo \log \log X)$.
\end{proof}

% ============================================================
\section{Prime Constellations and Consecutive Gap Vectors}
\label{sec:constellations}
% ============================================================

\subsection{Prime constellations}

Fix integers $0 = h_0 < h_1 < \cdots < h_m$. The corresponding \emph{prime constellation} event is
\[
\cC_{\mathbf{h}}(n) := \bigwedge_{j=0}^m (n + h_j \text{ is prime}).
\]

\subsubsection{Sieve upper bound for $k$-tuples}

\begin{theorem}[Selberg/Brun upper bound for prime $k$-tuples]\label{thm:ktuple-sieve}
Fix $m \geq 1$. There exists a constant $C_m > 0$ such that for all $X \geq 3$ and all integer $m$-tuples $\mathbf{h} = (h_0, \ldots, h_m)$ with $0 = h_0 < \cdots < h_m \leq CX$,
\[
A_{\mathbf{h}}(X) := \#\{n \in [X, 2X] : n + h_0, \ldots, n + h_m \in \PP\} \leq C_m \mathfrak{S}(\mathbf{h}) \frac{X}{(\log X)^{m+1}},
\]
where $\mathfrak{S}(\mathbf{h}) \geq 0$ is the standard $k$-tuple singular series. Moreover, $\mathfrak{S}(\mathbf{h}) \ll_m (\log \log(3X))^{O_m(1)}$.
\end{theorem}

\subsubsection{Fixed constellations are rare}

\begin{theorem}[A fixed constellation has small probability]\label{thm:constellation-fixed}
Fix $m \geq 1$ and offsets $0 = h_0 < h_1 < \cdots < h_m \leq 2X$. Then for $X$ large,
\[
\Pr_{P \in \cP_X}[\cC_{\mathbf{h}}(P)] \leq C_m \frac{\mathfrak{S}(\mathbf{h})}{(\log X)^m} \ll_m \frac{(\log \log(3X))^{O_m(1)}}{(\log X)^m}.
\]
\end{theorem}

\begin{proof}[Proof sketch]
The number of primes $p \in [X, 2X]$ satisfying the constellation event $\cC_{\mathbf{h}}(p)$ is at most $A_{\mathbf{h}}(X)$, since if $p, p+h_1, \ldots, p+h_m$ are all prime, then $p$ is counted in $A_{\mathbf{h}}(X)$. By Theorem~\ref{thm:ktuple-sieve}:
\[
A_{\mathbf{h}}(X) \leq C_m \mathfrak{S}(\mathbf{h}) \frac{X}{(\log X)^{m+1}}.
\]
Dividing by the number of primes in $[X, 2X]$, which satisfies $\pi_X \gg X / \log X$ by the prime number theorem, yields
\[
\Pr_{P \in \cP_X}[\cC_{\mathbf{h}}(P)] \leq \frac{A_{\mathbf{h}}(X)}{\pi_X} \ll \frac{C_m \mathfrak{S}(\mathbf{h}) X / (\log X)^{m+1}}{X / \log X} = C_m \frac{\mathfrak{S}(\mathbf{h})}{(\log X)^m}.
\]
The crude bound $\mathfrak{S}(\mathbf{h}) \ll_m (\log \log(3X))^{O_m(1)}$ from Theorem~\ref{thm:ktuple-sieve} completes the proof.
\end{proof}

\begin{corollary}[Sparse-set anti-concentration for constellations]\label{cor:constellation-sparse}
Fix $m \geq 1$ and consider a finite collection $\mathcal{H}$ of offset patterns. Then for $X$ large,
\[
\Pr_{P \in \cP_X}[\exists \mathbf{h} \in \mathcal{H} : \cC_{\mathbf{h}}(P)] \ll_m \frac{|\mathcal{H}| \cdot (\log \log(3X))^{O_m(1)}}{(\log X)^m}.
\]
\end{corollary}

\begin{proof}[Proof sketch]
Apply a union bound over patterns $\mathbf{h} \in \mathcal{H}$:
\[
\Pr_{P \in \cP_X}[\exists \mathbf{h} \in \mathcal{H} : \cC_{\mathbf{h}}(P)] \leq \sum_{\mathbf{h} \in \mathcal{H}} \Pr_{P \in \cP_X}[\cC_{\mathbf{h}}(P)].
\]
By Theorem~\ref{thm:constellation-fixed}, each term is $\ll_m (\log \log(3X))^{O_m(1)} / (\log X)^m$. Summing over $|\mathcal{H}|$ patterns yields the result.
\end{proof}

\subsection{Consecutive gap vectors}

Fix $t \geq 1$. For a prime $p$, define the length-$t$ consecutive gap vector
\[
\mathbf{G}^{(t)}(p) := (g(p), g(p^+), \ldots, g(p^{+(t-1)})) \in \N^t.
\]
Let $\mu_X^{(t)}$ denote the law of $\mathbf{G}^{(t)}_X := \mathbf{G}^{(t)}(P)$ for $P$ uniform in $\cP_X$.

For $\mathbf{g} = (g_1, \ldots, g_t)$ define cumulative offsets $H_0 := 0$, $H_k := g_1 + \cdots + g_k$.

\begin{theorem}[Any fixed consecutive $t$-gap vector is rare]\label{thm:consecutive-fixed}
Fix $t \geq 1$. For all $X$ large and all $\mathbf{g} = (g_1, \ldots, g_t) \in \N^t$,
\[
\mu_X^{(t)}(\mathbf{g}) \leq C_t \frac{\mathfrak{S}(0, H_1, \ldots, H_t)}{(\log X)^t} \ll_t \frac{(\log \log(3X))^{O_t(1)}}{(\log X)^t}.
\]
\end{theorem}

\begin{proof}
If $\mathbf{G}^{(t)}(p) = \mathbf{g}$, then $p + H_0, \ldots, p + H_t$ are all prime. Apply Theorem~\ref{thm:ktuple-sieve} with the $(t+1)$-tuple $(H_0, \ldots, H_t)$ and divide by $\pi_X \gg X / \log X$.
\end{proof}

\begin{theorem}[Menu bound for consecutive $t$-gap patterns]\label{thm:menu-consecutive}
Fix $t \geq 1$. For all $X$ large and all finite sets $\mathcal{S} \subseteq \N^t$,
\[
\Pr[\mathbf{G}^{(t)}_X \in \mathcal{S}] \leq C_t \frac{|\mathcal{S}| \cdot (\log \log(3X))^{O_t(1)}}{(\log X)^t}.
\]
\end{theorem}

\begin{proof}[Proof sketch]
By the union bound:
\[
\Pr[\mathbf{G}^{(t)}_X \in \mathcal{S}] = \sum_{\mathbf{g} \in \mathcal{S}} \mu_X^{(t)}(\mathbf{g}) \leq |\mathcal{S}| \cdot \max_{\mathbf{g} \in \mathcal{S}} \mu_X^{(t)}(\mathbf{g}).
\]
By Theorem~\ref{thm:consecutive-fixed}, the maximum is $\ll_t (\log \log(3X))^{O_t(1)} / (\log X)^t$, yielding the result.
\end{proof}

\begin{corollary}[Collision probability for consecutive patterns]\label{cor:collision-consecutive}
Fix $t \geq 1$ and define
\[
w_X^{(t)} := \sum_{\mathbf{g} \in \N^t} \mu_X^{(t)}(\mathbf{g})^2.
\]
Then for $X$ large,
\[
w_X^{(t)} \ll_t \frac{(\log \log(3X))^{O_t(1)}}{(\log X)^t}, \qquad \ell_X^{(t)} := 1 - w_X^{(t)} \geq 1 - \frac{(\log \log(3X))^{O_t(1)}}{(\log X)^t}.
\]
\end{corollary}

\begin{proof}[Proof sketch]
As in the single-gap case (Theorem~\ref{thm:collision-single}), we use:
\[
w_X^{(t)} = \sum_{\mathbf{g}} \mu_X^{(t)}(\mathbf{g})^2 \leq \left(\max_{\mathbf{g}} \mu_X^{(t)}(\mathbf{g})\right) \cdot \sum_{\mathbf{g}} \mu_X^{(t)}(\mathbf{g}) = \max_{\mathbf{g}} \mu_X^{(t)}(\mathbf{g}),
\]
since $\sum_{\mathbf{g}} \mu_X^{(t)}(\mathbf{g}) = 1$. Theorem~\ref{thm:consecutive-fixed} bounds the maximum, and the logical entropy bound follows immediately.
\end{proof}

% ============================================================
\section{Compositionality and Amplification}
\label{sec:compositionality}
% ============================================================

One conceptual advantage of the logical-entropy viewpoint is that collision probabilities tensorize under independent sampling.

\subsection{Tensorization of collisions}

\begin{proposition}[Tensorization for collision probabilities]\label{prop:tensorize}
Let $Z$ be a discrete random variable with law $\nu$, and define $w(\nu) = \sum_z \nu(z)^2$. Let $Z^{(1)}, \ldots, Z^{(m)}$ be i.i.d.\ copies of $Z$. Then
\[
\Pr[(Z^{(1)}, \ldots, Z^{(m)}) = (\widetilde{Z}^{(1)}, \ldots, \widetilde{Z}^{(m)})] = (w(\nu))^m.
\]
\end{proposition}

\begin{proof}[Proof sketch]
Let $(\widetilde{Z}^{(1)}, \ldots, \widetilde{Z}^{(m)})$ be an independent copy of the $m$-tuple. The two tuples coincide if and only if each coordinate coincides: $Z^{(i)} = \widetilde{Z}^{(i)}$ for all $i$. By independence of coordinates and independence of the two copies:
\[
\Pr[(Z^{(1)}, \ldots, Z^{(m)}) = (\widetilde{Z}^{(1)}, \ldots, \widetilde{Z}^{(m)})] = \prod_{i=1}^m \Pr[Z^{(i)} = \widetilde{Z}^{(i)}] = \prod_{i=1}^m w(\nu) = (w(\nu))^m.
\]
\end{proof}

\begin{corollary}[Amplification for prime gaps]\label{cor:amplify}
Let $G_X^{(1)}, \ldots, G_X^{(m)}$ be i.i.d.\ draws from $\mu_X$. Then
\[
\Pr[(G_X^{(1)}, \ldots, G_X^{(m)}) = (\widetilde{G}_X^{(1)}, \ldots, \widetilde{G}_X^{(m)})] = (w_X)^m \ll \left(\frac{\log \log(3X)}{\log X}\right)^m.
\]
\end{corollary}

\begin{proof}[Proof sketch]
Apply Proposition~\ref{prop:tensorize} with $\nu = \mu_X$ and $w(\nu) = w_X$. The bound $(w_X)^m \ll ((\log \log(3X))/(\log X))^m$ follows from Theorem~\ref{thm:collision-single}.
\end{proof}

\subsection{Menus fail exponentially under repetition}

\begin{proposition}[Menus fail exponentially]\label{prop:menu-repeat}
Let $Z$ be a discrete random variable and let $S$ be any set with $\Pr[Z \in S] \leq \alpha$. If $Z^{(1)}, \ldots, Z^{(m)}$ are i.i.d., then
\[
\Pr[Z^{(1)} \in S, \ldots, Z^{(m)} \in S] \leq \alpha^m.
\]
\end{proposition}

\begin{proof}[Proof sketch]
By independence, the probability that all $m$ samples fall in $S$ factors:
\[
\Pr[Z^{(1)} \in S, \ldots, Z^{(m)} \in S] = \prod_{i=1}^m \Pr[Z^{(i)} \in S] = (\Pr[Z \in S])^m \leq \alpha^m.
\]
Combined with Theorem~\ref{thm:menu-single} or Theorem~\ref{thm:menu-consecutive}, this shows that even if a ``menu'' $S$ has small but nonzero success probability on one sample, repeating across independent primes drives success probability down exponentially.
\end{proof}

% ============================================================
\section{Sieve-Theoretic Framework: Local to Selberg Weakness}
\label{sec:sieve-framework}
% ============================================================

We now develop a sieve-native weakness notion that makes the resulting theorems as clean as PSI-W. The idea is to define weakness directly in terms of local congruence obstructions, which are exactly what sieves control.

This framework can be understood as an instantiation of quantale weakness (Section~\ref{subsec:quantale-weakness}) in a multiplicative setting. The local weakness $\WeakLoc$ is an Euler product, which corresponds to quantale weakness in a product quantale where the monoidal operation $\otimes$ is ordinary multiplication and the valuation encodes local density factors $(1 - \omega(p)/(p-1))$ at each prime $p$. The Selberg weakness $\WeakSel$ then arises as an optimized upper bound on this product-quantale weakness, computed via the Selberg sieve's quadratic form.

\subsection{Local obstruction patterns}

Fix parameters: a large scale $X \geq 3$, a wheel modulus $W_0 \geq 1$ and a residue $b \pmod{W_0}$ with $(b, W_0) = 1$, and a sieve cutoff $z \geq 2$.

\begin{definition}[Local obstruction pattern]\label{def:pattern}
A \emph{local obstruction pattern up to $z$} is a choice, for each prime $p \leq z$ with $p \nmid W_0$, of a set of forbidden \emph{reduced} residue classes
\[
\Omega(p) \subseteq (\Z/p\Z)^\times.
\]
Write $\omega(p) := |\Omega(p)|$, with $0 \leq \omega(p) \leq p - 1$.
\end{definition}

Given $\Omega$, define the sifted set:
\[
\cS_\Omega(X; W_0, b, z) := \{n \in [X, 2X] \cap \Z : n \equiv b \pmod{W_0}, \; n \bmod p \notin \Omega(p) \; \forall p \leq z, p \nmid W_0\}.
\]

\subsection{Local weakness}

\begin{definition}[Local weakness]\label{def:weakloc}
The \emph{local weakness} of an obstruction pattern $\Omega$ up to sieve level $z$ is defined by
\[
\WeakLoc(\Omega; z) := \prod_{\substack{p \leq z \\ p \nmid W_0}} \left(1 - \frac{\omega(p)}{p-1}\right).
\]
\end{definition}

This is an instance of quantale weakness in the following sense. Consider the multiplicative quantale $((0,1], \times, \max, 1, \leq)$. For each prime $p \leq z$, define the local valuation $\mu(p) := 1 - \omega(p)/(p-1)$, representing the fraction of reduced residue classes that survive sieving at $p$. The local weakness is then
\[
\WeakLoc(\Omega; z) = \bigotimes_{p \leq z, \, p \nmid W_0} \mu(p) = \prod_{p} \mu(p),
\]
which measures the ``total survival fraction'' across all local obstructions. High local weakness means many residue classes survive (the sieve is weak/permissive); low local weakness means few survive (the sieve is strong/restrictive).

\begin{proposition}[Multiplicativity for disjoint ranges]\label{prop:mult}
Let $\Omega_1, \Omega_2$ be obstruction patterns supported on disjoint prime sets. Let $\Omega = \Omega_1 \cup \Omega_2$. Then
\[
\WeakLoc(\Omega; z) = \WeakLoc(\Omega_1; z) \cdot \WeakLoc(\Omega_2; z).
\]
\end{proposition}

\begin{proof}[Proof sketch]
Let $P_1 = \{p \leq z : \Omega_1(p) \neq \emptyset\}$ and $P_2 = \{p \leq z : \Omega_2(p) \neq \emptyset\}$ be the supporting prime sets, assumed disjoint. For the combined pattern $\Omega = \Omega_1 \cup \Omega_2$, we have $\omega(p) = \omega_1(p)$ if $p \in P_1$, $\omega(p) = \omega_2(p)$ if $p \in P_2$, and $\omega(p) = 0$ otherwise. Thus the Euler product factors:
\[
\WeakLoc(\Omega; z) = \prod_{\substack{p \leq z \\ p \nmid W_0}} \left(1 - \frac{\omega(p)}{p-1}\right) = \prod_{p \in P_1} \left(1 - \frac{\omega_1(p)}{p-1}\right) \cdot \prod_{p \in P_2} \left(1 - \frac{\omega_2(p)}{p-1}\right),
\]
which equals $\WeakLoc(\Omega_1; z) \cdot \WeakLoc(\Omega_2; z)$.
\end{proof}

This multiplicativity is a direct consequence of the quantale structure: weakness in product quantales decomposes over independent factors.

\subsection{Selberg weakness}

For squarefree $d$ with prime factors $\leq z$ and $(d, W_0) = 1$, define $\omega(d) := \prod_{p \mid d} \omega(p)$ and
\[
g(d) := \frac{\omega(d)}{\vp(d)}.
\]

\begin{definition}[Selberg weakness]\label{def:weaksel}
The \emph{Selberg quadratic form} for coefficients $\lambda = \{\lambda_d\}$ supported on squarefree $d \leq R$ is
\[
\cJ(\lambda) := \sum_{d_1} \sum_{d_2} \lambda_{d_1} \lambda_{d_2} g([d_1, d_2]).
\]
The \emph{Selberg weakness} is the minimum of this quadratic form subject to normalization:
\[
\WeakSel(\Omega; z, R) := \inf_{\lambda : \lambda_1 = 1} \cJ(\lambda).
\]
\end{definition}

The Selberg weakness can be understood as an optimized version of local weakness: it represents the tightest upper bound on the sifted count achievable by Selberg's method. In the quantale framework, $\WeakSel$ arises from optimizing over all ``weightings'' $\lambda$ of the local factors, finding the combination that yields the smallest provable bound while respecting the constraint $\lambda_1 = 1$ (which ensures the sieve upper-bounds the characteristic function of the sifted set).

\begin{proposition}[Local $\Rightarrow$ Selberg]\label{prop:local-to-selberg}
Assume $\omega(p) \leq \kappa$ for all $p \leq z$. Let $s := (\log R)/(\log z)$. Then for $s > 1$:
\[
\WeakSel(\Omega; z, R) \leq (F_\kappa(s) + o_{s,\kappa}(1)) \WeakLoc(\Omega; z),
\]
where $F_\kappa(s)$ is the standard upper-bound sieve function of dimension $\kappa$.
\end{proposition}

\begin{proof}[Proof sketch]
This is a standard result from sieve theory (see \cite{HR, IK}). The key idea is that the multiplicative function $g(d) = \omega(d)/\vp(d)$ satisfies the ``dimension $\kappa$'' condition: for primes $p$, we have $g(p) = \omega(p)/(p-1) \leq \kappa/(p-1)$. 

The Selberg sieve constructs explicit coefficients $\lambda_d$ (supported on squarefree $d \leq R$) that minimize $\cJ(\lambda)$ subject to $\lambda_1 = 1$. The optimum satisfies
\[
\WeakSel(\Omega; z, R) = \cJ(\lambda^*) \leq \frac{1}{\sum_{d \leq R} \mu(d)^2 / g(d)} \cdot (1 + o(1)),
\]
and the denominator is asymptotic to $\WeakLoc(\Omega; z)^{-1} / F_\kappa(s)$ by standard estimates on sums of multiplicative functions. Rearranging gives the claimed bound.
\end{proof}

\subsection{Prime sieve theorem}

\begin{definition}[Level-of-distribution hypothesis]\label{def:LD}
Fix $\theta \in (0, 1]$. We say \emph{$\Lam$ has level of distribution $\theta$} if for every $A > 0$ there exists $B = B(A)$ such that for all $Q \leq X^\theta / (\log X)^B$:
\[
\sum_{q \leq Q} \max_{(a,q)=1} \left|\sum_{\substack{X < n \leq 2X \\ n \equiv a \pmod{q}}} \Lam(n) - \frac{X}{\vp(q)}\right| \ll_A \frac{X}{(\log X)^A}.
\]
\end{definition}

Bombieri--Vinogradov gives this with $\theta = 1/2$ unconditionally.

\begin{theorem}[Crisp prime upper bound in terms of local weakness]\label{thm:prime-sieve-loc}
Assume Definition~\ref{def:LD} holds at level $\theta$. Assume $\omega(p) \leq \kappa$ for all $p \leq z$. Let $R = z^s$ with $s > 1$ and suppose $W_0 R^2 \leq X^\theta / (\log X)^B$ for sufficiently large $B$. Then
\[
\pi_\Omega(X; W_0, b, z) \leq (F_\kappa(s) + o(1)) \WeakLoc(\Omega; z) \cdot \frac{X}{\vp(W_0) \log X} + O_A\left(\frac{X}{(\log X)^{A+1}}\right).
\]
\end{theorem}

\begin{proof}[Proof sketch]
Let $\lambda = \{\lambda_d\}$ be any admissible Selberg coefficients with $\lambda_1 = 1$. Define the Selberg square weight
\[
W_\lambda(n) := \left(\sum_{d \leq R} \lambda_d \1_{n \bmod d \in \Omega(d)}\right)^2.
\]
For $n \in \cS_\Omega(X; W_0, b, z)$, only the $d = 1$ term survives in the inner sum (since $n$ avoids all obstructions), so $W_\lambda(n) = 1$. Thus
\[
\1_{\{n \in \cS_\Omega\}} \leq \1_{\{n \equiv b \pmod{W_0}\}} \cdot W_\lambda(n).
\]
Summing $\Lam(n)$ (the von Mangoldt function) over $n \in [X, 2X]$:
\[
\sum_{\substack{n \in [X, 2X] \\ n \in \cS_\Omega}} \Lam(n) \leq \sum_{\substack{n \in [X, 2X] \\ n \equiv b \pmod{W_0}}} \Lam(n) W_\lambda(n).
\]
Expanding the square and using the level-of-distribution hypothesis to evaluate sums of $\Lam(n)$ over arithmetic progressions modulo $q = W_0 [d_1, d_2] \leq W_0 R^2$, the main term becomes
\[
\frac{X}{\vp(W_0)} \sum_{d_1, d_2} \lambda_{d_1} \lambda_{d_2} g([d_1, d_2]) = \frac{X}{\vp(W_0)} \cJ(\lambda).
\]
Minimizing over $\lambda$ gives $\WeakSel(\Omega; z, R)$. Combining with Proposition~\ref{prop:local-to-selberg} and dividing by $\log X$ (since $\Lam(p) \geq \log X$ for primes $p \in [X, 2X]$) yields the result.
\end{proof}

\begin{corollary}[Independent local filters multiply densities]\label{cor:comp}
Let $\Omega_1, \Omega_2$ be patterns supported on disjoint prime sets, $\Omega = \Omega_1 \cup \Omega_2$. Then
\[
\pi_\Omega(X; W_0, b, z) \ll \WeakLoc(\Omega_1; z) \WeakLoc(\Omega_2; z) \cdot \frac{X}{\vp(W_0) \log X}.
\]
\end{corollary}

\begin{proof}[Proof sketch]
By Proposition~\ref{prop:mult}, $\WeakLoc(\Omega; z) = \WeakLoc(\Omega_1; z) \cdot \WeakLoc(\Omega_2; z)$ when the supporting prime sets are disjoint. Apply Theorem~\ref{thm:prime-sieve-loc} to the combined pattern $\Omega$, absorbing the sieve constant $F_\kappa(s)$ into the implied constant.
\end{proof}

% ============================================================
\section{Discussion}
\label{sec:discussion}
% ============================================================

\subsection{Compatibility with known results}

All PSI conjectures and the proven weakness results are compatible with what is currently known about the distribution of primes and with standard conjectures on prime gaps.

The Prime Number Theorem controls the average density of primes but does not provide an algorithm for computing the successor of a given prime without testing candidates for primality. Conjectures such as Cram\'er's conjecture imply upper bounds of the form $g_n \leq (\log p_n)^{2+o(1)}$; PSI-T then asserts that no algorithm can compute the prime successor asymptotically faster than the sequential method by more than a polylogarithmic factor.

Conjectures such as the twin prime conjecture concern infinitely many indices $n$ for which $g_n$ is small. This does not contradict any PSI variant, as all of them allow for some inputs where the next prime is unusually close.

\subsection{Plausibility from the Cram\'er model}

In the Cram\'er model, primes are modeled by independent random variables where each integer $n$ is declared ``prime'' with probability $1/\log n$. Under this heuristic, gaps near height $X$ are approximately geometric with mean $\Theta(\log X)$, and a typical gap's binary encoding has length $\Theta(\log \log X)$. One expects $\K(\code(G) \mid X) \approx \log \log X$ for most samples $G$, which is exactly PSI-K.

\subsection{What the results prove vs.\ conjecture}

The proven results (PSI-K for $c < 1$, PSI-W, PSI-W-LE) show strong anti-concentration and high entropy properties. They are consistent with strong conjectural models but do not approach them in strength:
\begin{itemize}[leftmargin=2em]
\item We do \emph{not} prove Poisson statistics for normalized gaps.
\item We do \emph{not} prove independence of consecutive gaps.
\item We only use sieve \emph{upper bounds}, so our statements are inherently ``one-sided.''
\end{itemize}
Nevertheless, PSI-W and PSI-W-LE provide clean, compositional, and provable proxies for ``randomness of prime spacings'' at the level of diversity and entropy.

\subsection{The $c = 1$ barrier}

Our unconditional proof of PSI-K$(c, \delta)$ works for all fixed $c < 1$ but does not extend to $c = 1$. This is because the union bound over low-complexity gap values becomes ineffective when $|S_c(X)| \approx \log X$. Proving an additive incompressibility statement of the form $\K(\code(g(p)) \mid X) \geq \logtwo \logtwo X - O(1)$ for most gaps would require non-concentration statements stronger than bounding each fixed gap separately.

% ============================================================
\section{Conclusion}
% ============================================================

We have developed multiple complementary formalizations of Prime Successor Irreducibility:

\begin{enumerate}[leftmargin=2em]
\item \textbf{PSI-T} asserts running-time lower bounds for prime successor algorithms.
\item \textbf{PSI-K} asserts incompressibility of typical prime gaps given the scale; we prove PSI-K$(c, \delta)$ unconditionally for all $c < 1$.
\item \textbf{PSI-W} shows that no small menu of gap values can capture significant probability mass.
\item \textbf{PSI-W-LE} shows that collision probability decays, hence logical entropy tends to $1$.
\item The \textbf{local-to-Selberg weakness framework} provides compositional sieve bounds.
\end{enumerate}

These results extend to prime constellations and consecutive gap vectors. The combination of operational and informational formulations offers a promising template for connecting computational irreducibility with classical questions in number theory.

The unifying thread through the weakness formulations is the quantale-weakness framework introduced in Section~\ref{subsec:quantale-weakness}. By recognizing collision probability, menu bounds, and sieve-theoretic weakness as instances of a single algebraic structure, we gain both conceptual clarity and compositional tools for combining and extending these results.

Overall, while the results of this paper are technically involved and framed in statistical and sieve-theoretic terms, they conceptually support the original motivating conjecture articulated by Lauritzen, which was the inspiration for the lines of investigation taken: that despite large-scale regularities in the distribution of primes, the task of identifying the immediate successor of a given prime remains locally irreducible. That is, even when global structure is known, the information needed to determine the next prime is typically not compressible beyond trivial scale information. PSI-T expresses the strongest algorithmic form of this intuition, while PSI-K and PSI-W provide rigorous, unconditional results that capture provably accessible consequences of it.

% ============================================================
\appendix

\section{Bertrand's Postulate and Range Bounds}\label{app:bertrand}

\begin{lemma}[Iterated Bertrand range bound]\label{lem:bertrand}
For every integer $n \geq 2$, there exists a prime in $(n, 2n)$. In particular, if $p$ is prime then $p^+ \leq 2p$. More generally, for fixed $t \geq 1$ and any prime $p$:
\[
p^{+(t)} \leq 2^t p.
\]
\end{lemma}

\section{A Crude Bound on the Prime-Pair Singular Series}\label{app:ss}

\begin{lemma}[Crude bound for the pair singular series]\label{lem:ss-crude}
Let $F(h) := \prod_{p \mid h, p > 2} \frac{p-1}{p-2}$. Then for $h \geq 3$, $F(h) \ll \log \log(3h)$.
\end{lemma}

\begin{proof}
Write
\[
\log F(h) = \sum_{\substack{p \mid h \\ p > 2}} \log\left(1 + \frac{1}{p-2}\right) \leq \sum_{\substack{p \mid h \\ p > 2}} \frac{2}{p} \leq 2 \sum_{p \leq y} \frac{1}{p},
\]
where $y$ is the largest prime factor scale. Since $\prod_{p \leq y} p \leq h$, one has $y \ll \log h$. Using $\sum_{p \leq y} 1/p = \log \log y + O(1)$ gives $\log F(h) \ll \log \log \log(3h)$, hence $F(h) \ll \log \log(3h)$.
\end{proof}

\section{Bit-Complexity of Primality Testing}\label{app:primality}

We provide concrete estimates for the primality cost function $C(k)$ referenced in Section~\ref{sec:PSI-T}.

\subsection{Miller-Rabin with deterministic cleanup}

A standard approach combines probabilistic Miller-Rabin rounds with a deterministic backup:

\begin{enumerate}[leftmargin=2em]
\item Run $r$ rounds of Miller-Rabin with fixed, deterministic bases. Each round costs $O(k^3)$ bit operations (using schoolbook multiplication) or $O(k^2 \log k \log \log k)$ with fast multiplication.
\item If all rounds declare ``probable prime,'' run a deterministic test such as AKS.
\end{enumerate}

For AKS, the worst-case cost is $D(k) = O(k^{c_0})$ for some explicit constant $c_0$ (early versions had $c_0 \approx 12$; refinements achieve $c_0 \approx 6$). The combined algorithm has
\[
C(k) = O(k^{\max(3, c_0)}) = O(k^c)
\]
for some constant $c \geq 3$.

\subsection{Implications for PSI-T}

With this estimate, the sequential algorithm's cost on input $p_n$ satisfies
\[
T_{\mathrm{Seq}}(p_n) \asymp g_n \cdot C(\ell(p_n)) = O(g_n \cdot \ell(p_n)^c).
\]

The PSI-T conjectures (Section~\ref{sec:PSI-T}) assert that for any prime successor algorithm $M$:
\[
\mathrm{time}_M(p_n) \geq \frac{c'}{(\log p_n)^\alpha} \cdot g_n \cdot C(\ell(p_n))
\]
for some constants $c' > 0$ and $\alpha > 0$, holding on the specified set of primes (infinitely many, positive density, density one, or all but finitely many, depending on the variant).

The key point is that the lower bound scales with $g_n$---the number of integers in the gap---not merely with the bit-length $\ell(p_n)$. This captures the intuition that one cannot avoid examining most candidates in the gap.

% ============================================================

\end{document}